\documentclass[conference,10pt]{IEEEtran}
\usepackage{bm,cite,float,amsmath,amssymb,amsthm}

\usepackage{algorithm}
\usepackage[noend]{algpseudocode}
\usepackage{indentfirst}

\makeatletter
\def\BState{\State\hskip-\ALG@thistlm}
\makeatother

\DeclareMathOperator\erfc{erfc}

\usepackage{amssymb}
\usepackage{amsmath}
\usepackage{graphicx}
\usepackage{cite}
\usepackage{citesort}
\usepackage{balance}
\usepackage[utf8]{inputenc}

\usepackage{graphicx, epstopdf}
\usepackage{setspace}
\usepackage{acronym}
\usepackage{mathrsfs}
\usepackage{ifthen}
\usepackage{textcomp}
\usepackage{float}
\usepackage{subfigure}
\usepackage{color}
\usepackage{graphicx, epstopdf}
\usepackage{setspace}
\usepackage{graphicx,cite,amsmath,amssymb,hhline}
\usepackage{lineno}

\usepackage{array}
\usepackage{mdwmath}
\usepackage{mdwtab}
\usepackage{verbatim}
\usepackage{color}
\usepackage{multirow}

\newtheorem{cor}{Corollary}
\newtheorem{ppro}{Proposition}

\IEEEoverridecommandlockouts

\usepackage{graphicx,epstopdf}
\usepackage{epsfig}	
\usepackage{amsfonts,balance}
\usepackage{bbm}
\floatname{algorithm}{Algorithm}
\usepackage{lipsum}

\newtheorem{theorem}{Theorem}

\newtheorem{corollary}{Corollary}

\usepackage[
top    = 1.70cm,
bottom = 1.05in,
left   = 0.63 in,
right  = 0.63 in]{geometry}

\makeatletter
\def\ScaleIfNeeded{%
\ifdim\Gin@nat@width>\linewidth \linewidth \else \Gin@nat@width
\fi } \makeatother

\begin{document}

\title{  3D Stochastic Geometry Model for Large-Scale  Molecular Communication Systems}

\author{
\IEEEauthorblockN{ Yansha Deng\IEEEauthorrefmark{1}, Adam~Noel\IEEEauthorrefmark{2}, Weisi~Guo\IEEEauthorrefmark{3}, 
Arumugam~Nallanathan\IEEEauthorrefmark{1}, 
and  Maged~Elkashlan\IEEEauthorrefmark{4}  
%and
%Robert Schober\IEEEauthorrefmark{2}\IEEEauthorrefmark{4}
 } \IEEEauthorblockA{
\IEEEauthorrefmark{1}Department of Informatics, King's College London,  UK\\
\IEEEauthorrefmark{2} School of Electrical Engineering and Computer Science, University of Ottawa, Canada \\
\IEEEauthorrefmark{3} School of Engineering, University of Warwick, UK\\
\IEEEauthorrefmark{4} School of Electronic Engineering and Computer
Science, Queen Mary University of London, UK\\
%\IEEEauthorrefmark{4}Institute for Digital Communication, Friedrich-Alexander-Universit\"{a}t Erlangen-N\"{u}rnberg (FAU), Erlangen, Germany\\
 } }

%{$^{\ast}$Department of Computer Science and Operations Research\\
%	Universit\'{e} de Montr\'{e}al\\
%	Email: adamnoel@umontreal.ca
%	\\ $^{\dagger}$School of Electrical Engineering and Computer Science\\
%	University of Ottawa}

\newcommand{\meter}{\textnormal{m}}
\newcommand{\micron}{\mu\textnormal{m}}
\newcommand{\second}{\textnormal{s}}

% Changes in newer revision
%\newcommand{\edit}[2]{\textbf{#1}}
\newcommand{\edit}[2]{#1}

\maketitle
\vspace{-3cm}

\begin{abstract}
Information delivery using chemical molecules is an integral part of biology at multiple distance scales and has attracted recent interest in bioengineering and communication.  The  collective signal strength at the receiver (i.e., the expected number of observed molecules inside the receiver), resulting from a large number of transmitters at random distances (e.g., due to mobility), can have a major impact on the reliability and efficiency of the molecular communication system.  Modeling the collective signal from multiple diffusion sources can be computationally and analytically challenging. In this paper, we present the first tractable analytical model for the collective signal strength due to randomly-placed transmitters, whose positions are modelled as a homogeneous Poisson point process in three-dimensional (3D) space. By applying  stochastic geometry, we derive   analytical expressions for the expected  number of observed molecules at a fully absorbing receiver and a passive receiver. Our results reveal that the  collective signal strength at both types of receivers increases proportionally with increasing transmitter density.  The proposed  framework dramatically simplifies the  analysis of large-scale molecular systems in both communication and biological applications.

%In this report, we model the chemical signal strength from $K$ molecular messenger nodes.  With the aid of stochastic geometry, we show that if the messenger nodes' location follow a uniform random spatial distribution, the received signal strength statistics from one or more nodes can be represented by closed form expressions.  We then prove that in the absence of individual synchronisation between a pair of messengers, random asynchronous transmission is superior to all the messengers sharing a common preset transmission clock.
\end{abstract}


\section{Introduction}

Molecular communication via diffusion has attracted significant bioengineering and communication engineering research interest in recent years \cite{EckfordBook13}.  Messages are delivered via molecules undergoing random walks \cite{Codling08}, a prevalent phenomenon in biology \cite{Atkinson09}. In fact, molecular communication exists in nature at both the nano- and macro-scales, offering transmit energy and signal propagation advantages over wave-based communications \cite{Llatser13,Guo15TMBMC}. One application example is that  swarms of nano-robots can track specific targets, such as a tumour cells, to perform operations such as targeted drug delivery \cite{Douglas12}.  In order to do so, energy efficient and tether-less communications between the nano-robots must be established in biological conditions \cite{Cavalcanti06}, and possibly additional nano-bio-interfaces need to be implemented \cite{Kirkpatrick10}. 
%\cite{Wyatt09}
% Examples include a reverse-engineered pheromone signalling system \cite{Cole09} and a macro-scale prototype that can reliably send generic text messages using molecule \cite{Farsad13PLOS}.

Fundamentally, molecular communications involves modulating information onto the property of a single or a group of molecules (e.g.,  number, type, emission time). 
When modulating  the number of  molecules, each messenger node will transmit information-bearing molecules via chemical pulses. 
 In a realistic  environment with a swarm of robots (i.e., messenger nodes)  operating together, they are likely to transmit molecular messages simultaneously.  Due to limitations in transmitter design and molecule type availability, it is likely that many transmitters will transmit the same type of information molecule.
Thus,  it is important to model the  collective signal strength due to all transmitters with the same type of information molecule, and to account for random transmitter locations due to mobility.

Existing works have largely focused on  modeling:  1) the signal strength  of a point-to-point communication channel by considering the self-interference that arises from adjacent symbols (i.e., inter-symbol-interference (ISI)) at a passive receiver \cite{noel2014unify}, at a fully absorbing receiver \cite{yilmaz2014simulation}, and at a reversible adsorption receiver \cite{Yansha16}; and 2) the collective signal strength  of a multi-access communication channel at the passive receiver due to co-channel transmitters (i.e., transmitters emitting the same type of molecule) with the given knowledge of their total  number and location \cite{noel2014unify}.

The first work to  consider randomly distributed co-channel transmitters in 3-D space according  to a spatial homogeneous Poisson process (HPPP) is \cite{pierobon2014statistical}, where the probability density function (PDF) of the received signal at a point location was derived based on the assumption of  white Gaussian transmit signals.
%. In their work,  the transmit signal is assumed to be a white Gaussian signal, which leads the received signal under a given transmitter location to be a zero-mean stationary Gaussian process.
 Since the receiver size was negligible, the placement of transmitters did not need to accommodate the receiver's location. More importantly, only the  Monte Carlo simulation, and not particle-based simulation, was performed to verify the derived PDF.

From the perspective of receiver type, many works have focused on the passive receiver, which can observe and count the number of molecules inside the receiver without interfering with the molecules \cite{noel2014unify,pierobon2014statistical}. In nature, receivers commonly remove information molecules from the environment once they bind to a receptor. One example is the fully absorbing receiver, which absorbs all the molecules hitting its surface \cite{yilmaz2014simulation,Yansha16}. 
However, no work has studied the channel characteristics and the received signal at the fully absorbing receiver in a large-scale molecular communication system, let alone its comparison with 
that at the passive receiver.

In this paper, we model the collective signal strength at the passive receiver and fully absorbing receiver  due to  a swarm of mobile point transmitters that simultaneously emit a given number  of information molecules.   Unlike \cite{pierobon2014statistical}, which  focused on the  statistics of the received signal at any point location, we focus on examining and deriving  exact expressions for the expected number of molecules observed inside two types of receiver for signal demodulation. This is achieved 
using   stochastic geometry, which has been extensively  used to model and provide simple and tractable results for  wireless  systems \cite{baccelli2009stochastic}.  Our contributions can be summarized as follows:
\begin{enumerate}
\item We use stochastic geometry to model the collective signal at a receiver in a large-scale molecular communication system, where the receiver is either passive or fully absorbing. We distinguish between the desired signal   due to the nearest transmitter and the interfering  signal due to the other   transmitters.
\item We derive a simple closed-form expression for the expected  
number of  molecules absorbed at the fully absorbing
receiver, and a
 tractable expression for the expected number of 
molecules observed inside the passive receiver at any time
instant.
\item We  define and derive  tractable analytical expressions for   the fraction of molecules due to the nearest transmitter and the fraction of molecules due to the other transmitters.
\item We verify our results using  particle-based simulation and Monte Carlo simulation, which prove that the  expected number of molecules observed at both types of receiver increases linearly with increasing transmitter density.
\end{enumerate}

\section{System Model}

We consider a large-scale molecular communication system with a single receiver in which a swarm of point transmitters are  spatially distributed outside the receiver in  $\mathbb{R}^3/ {V_{{\Omega _{{r_r}}}}}$ according to an independent and HPPP $\Phi$ with density $\lambda$, where ${V_{{\Omega _{{r_r}}}}}$ is the volume of receiver ${{\Omega _{{r_r}}}}$.  This spatial distribution, which was previously used to model wireless sensor networks \cite{yan2016sensor},  cellular networks \cite{baccelli2009stochastic} and heterogenous cellular networks \cite{yan16hetnet},   has also been applied to model  bacterial colonies in \cite{Jeanson11} and the  interference sources in a  molecular communication system \cite{pierobon2014statistical}.  We consider a fluid environment  in the absence of flow currents: the extension for flow currents will be treated in future work.
 
 At any given time instant, a number of transmitters will be either silent or active. Thus, we define the activity probability of a transmitter that is triggered to transmit data  as ${\rho_a} (0<{\rho_a}<1)$. This activity probability  is independent of the receiver's location. Thus, the active point transmitters constitute independent HPPPs $\Phi_{a}$ with intensities $\lambda_a =\lambda \rho_{a}$. Each transmitter  transmits molecular signal pulses with amplitude $N_{\rm{tx}}^{\rm{FA}}$ ($N_{\rm{tx}}^{\rm{PS}}$)  to the absorbing receiver (the passive receiver).  
 We  assume  the existence of  a global clock such that  all molecule  emissions can only occur at $t=0$.
 
  We consider two types of spherical receiver with radius $r_r$: 1) Fully absorbing receiver \cite{Yilmaz14}, and 2) Passive receiver \cite{Llatser13,noel2014improving}. To equivalently compare them, we assume both types of receiver are capable of  counting the number of information molecules within the receiver volume at any time instant for information decoding.

%The fully absorbing receiver is capable of absorbing any messenger molecule hitting its surface, and it can count the  number of absorbed molecule in any arbitrary time interval for information decoding. In contrast, the passive receiver is transparent to  information molecule, and  is able to count and record the number of free molecule that are within the receiver volume in any time instance for information decoding, without interfering with their diffusion.

  %As shown in Fig.~\ref{System}, 

It is well known that the distance between the transmitter and the receiver  in molecular communication is the main contributor to the degradation of   the signal strength (i.e., the number of molecules observed at the receiver). Instinctively,
we assume that the receiver is associated with the nearest transmitter  to obtain the strongest signal.
Thus, the messenger molecules transmitted by other active point transmitters  act as  interference, which impairs the correct reception at the  receiver. To measure this impairment,  we formulate the desired signal, the interfering signal for  the absorbing receiver and  the passive receiver in the following subsections.

\subsection{Absorbing Receiver}

%In a point-to-point molecular communication system with a single point source located distance $d$ away from the closest point on the surface of a single absorbing receiver, the fraction of molecule absorbed at the receiver until time $t$, due to  the molecule emission occurring at $t=0$, is derived as \cite{Yilmaz14}
%\begin{align}
%F\left( {\left. {{\Omega _{{r_r}}},t} \right|d} \right) = \frac{{{r_r}}}{{d + {r_r}}}\erfc\Big\{ {\frac{d}{{\sqrt {4Dt} }}} \Big\}, \label{Fraction_d}
%\end{align}
%where $D$ is the constant diffusion coefficient, which is usually obtained via experiment as in \cite[Ch. 5]{cussler2009diffusion}.

%\subsection{Large-Scale Molecular Communication System}

 \begin{figure}[t]
	\centering
	\includegraphics[width=0.60\linewidth]{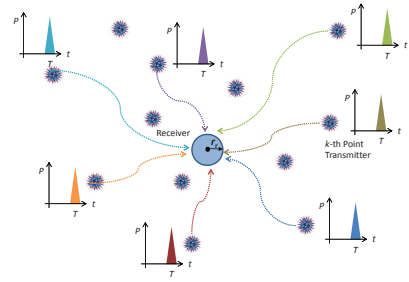}
	\caption{Illustration of a receptor receiving molecular pulse signals from point transmitters at different distances.}
	\label{System}
\end{figure}

In our proposed large-scale molecular communication system, let us consider the center of an absorbing receiver located at the origin. Using  the Slivnyak-Mecke's theorem \cite{baccelli2009stochastic}, the fraction $F^{\rm{FA}}$ of molecules absorbed at the  receiver until time $T$ due to an \emph{arbitrary} point transmitter $x$ at the location $\textbf{x}$ with  molecule emission occurring at $t=0$ can be represented  as \cite{Yilmaz14}
\begin{align}
F^{\rm{FA}}\left( {\left. {{\Omega _{{r_r}}},T} \right|\left\| \textbf{x} \right\|} \right) = \frac{{{r_r}}}{{\left\| \textbf{x} \right\| }}\erfc\Big\{ {\frac{{\left\| \textbf{x} \right\|-{r_r}}}{{\sqrt {4DT} }}}  \Big\}, \label{Fraction_x}
\end{align}
where $\left\| \textbf{x} \right\|$ is the distance between the point transmitter and the center of the  receiver where the transmitters follow a HPPP,
and  $D$ is the constant diffusion coefficient, which is usually obtained via experiment as in \cite[Ch. 5]{cussler2009diffusion}. 
%In practice, this distance $\left\| \textbf{x} \right\|$ is usually much  larger than $r_r$ (i.e., $\left\| \textbf{x} \right\|\gg r_r$).
%
%%\begin{figure}[t]
%%	\centering
%%	\includegraphics[width=0.60\linewidth]{System}
%%	\caption{Illustration of a receptor receiving molecular pulse signals from point transmitters at different distances.}
%%	\label{System}
%%\end{figure}
%%
%% 
%
The fraction $F_{\rm{u}}^{{\rm{FA}}}$ of  molecules absorbed inside the  receiver  until time $T$ due to a single  pulse emission by the \emph{nearest} active transmitter  can be represented as
\begin{align}
{F_{\rm{u}}^{{\rm{FA}}}\left( {\left. {{\Omega _{{r_r}}},T} \right|{{\left\| \textbf{x}^* \right\|}}} \right)}  = \frac{{{r_r}}}{{{{\left\| \textbf{x}^* \right\|}}}}\erfc\Big\{ {\frac{{{{\left\| \textbf{x}^* \right\|}} - {r_r}}}{{\sqrt {4DT} }}} \Big\}, \label{useful}
\end{align}
where $\left\| \textbf{x}^* \right\|$ denotes the distance between   the receiver and the nearest  transmitter,
\begin{align}\label{cell_selection_II}
 \quad x^{*} = \mathop {\arg \min }\limits_{x \in {\Phi _a}}  \left\| {{\textbf{x}}} \right\|, 
\end{align}
 $x^{*}$  denotes  the nearest point  transmitter  for the receiver, and ${\Phi _a}$ denotes the  set of active transmitters' positions.

The fraction  of  molecules absorbed at the  receiver until time $T$ due to  single  pulse emissions at each  active interfering transmitter  $F_{\rm{I}}^{{\rm{FA}}}$ and that due to a single  pulse emission at each active  transmitter $F_{\rm{all}}^{{\rm{FA}}}$ is represented as
\begin{align}
 {F_{\rm{I}}^{{\rm{FA}}}\left( {\left. {{\Omega _{{r_r}}},T} \right|\left\| \textbf{x} \right\|} \right)}  = \sum\limits_{{{{\Phi _a}} \mathord{\left/
 {\vphantom {{{\Phi _a}} {{\textbf{x}^*}}}} \right.
 \kern-\nulldelimiterspace} {{\textbf{x}^*}}}} {\frac{{{r_r}}}{{\left\| \textbf{x} \right\|}}\erfc\Big\{ {\frac{{\left\| \textbf{x} \right\| - {r_r}}}{{\sqrt {4DT} }}} \Big\}}, \label{interference}
\end{align}
and
\begin{align}
 {F_{\rm{all}}^{{\rm{FA}}}\left( {\left. {{\Omega _{{r_r}}},T} \right|\left\| \textbf{x} \right\|} \right)}  = \sum\limits_{{{{\Phi _a}} }} {\frac{{{r_r}}}{{\left\| \textbf{x} \right\|}}\erfc\Big\{ {\frac{{\left\| \textbf{x} \right\| - {r_r}}}{{\sqrt {4DT} }}} \Big\}}, \label{total_FA}
\end{align}
respectively.

The  expected   number of  molecules absorbed at the receiver by time $T$  due to  all active  transmitters is equivalently    the expected   number of  molecules absorbed at the receiver until time $T$, which can be   calculated as 
\begin{align}
{\mathbb{E}}&\left\{ {N_{{\rm{all}}}^{{\rm{FA}}}\left( { {{\Omega _{{r_r}}},T} } \right)} \right\}  
= N_{\rm{tx}}^{\rm{FA}} {F_{\rm{all}}^{{\rm{FA}}}\left( {\left. {{\Omega _{{r_r}}},T} \right|{{\left\| \textbf{x} \right\|}}} \right)} 
\nonumber\\ & =\underbrace {N_{{\rm{tx}}}^{{\rm{FA}}}F_{\rm{u}}^{{\rm{FA}}}\left( {\left. {{\Omega _{{r_r}}},T} \right|\left\| {{{\bf{x}}^*}} \right\|} \right)}_{{\mathbb{E}}_{\rm{u}}^{{\rm{FA}}}} +\underbrace {N_{{\rm{tx}}}^{{\rm{FA}}}F_{\rm{I}}^{{\rm{FA}}}\left( {\left. {{\Omega _{{r_r}}},T} \right|\left\| {{{\bf{x}}}} \right\|} \right)}_{{\mathbb{E}}_{\rm{I}}^{{\rm{FA}}}}
 , 
\end{align}
where $ {F_{\rm{all}}^{{\rm{FA}}}\left( {\left. {{\Omega _{{r_r}}},T} \right|{{\left\| \textbf{x} \right\|}}} \right)}$ is given in \eqref{total_FA}, 
%To reveal the impact of the signal to interference ratio (SIR) on system performance,  we define the SIR based on  the number of  molecules observed inside the fully absorbing receiver at time $T$  as 
%\begin{align}
%\rm{SIR}_{FA} = &  \frac{N_{\rm{tx}}^{\rm{FA}}{\mathbb{E}}\left\{{F_{\rm{u}}^{{\rm{FA}}}\left( {\left. {{\Omega _{{r_r}}},T} \right|{{\left\| \textbf{x}^* \right\|}}} \right)} \right\}  }{N_{\rm{tx}}^{\rm{FA}} {\mathbb{E}}\left\{{F_{\rm{I}}^{{\rm{FA}}}\left( {\left. {{\Omega _{{r_r}}},T} \right|\left\| \textbf{x} \right\|} \right)}\right\} } 
% = \frac{{{{\mathbb{E}}_{\rm{u}}^{{\rm{FA}}}}}}{{{{\mathbb{E}}_{\rm{I}}^{{\rm{FA}}}}}}, \label{SIRFA_raw} 
%\end{align}
%where the SIR  may dominate the bit error probability in a large-scale molecular communication system with binary concentration shift keying (CSK), 
 ${{\mathbb{E}}_{\rm{u}}^{{\rm{FA}}}}$ is the fraction of  absorbed molecules at the absorbing receiver until time $T$ due to the nearest transmitter, and ${{\mathbb{E}}_{\rm{I}}^{{\rm{FA}}}}$ is the fraction of  absorbed molecules at the absorbing receiver until time $T$ due to the other (interfering) transmitters.  
%The mathematical formulation of SIR can also be applied to   other types of receiver in large-scale molecular communication systems. 

\subsection{Passive Receiver}
In a point-to-point molecular communication system with a single point transmitter located distance $\left\| {\textbf{x}} \right\|$ away from  the center of a passive receiver with radius $r_r$, the local point concentration at the center of the passive receiver at time $T$ due to a single  pulse emission by the transmitter  is given as \cite[Eq. (4.28)]{nelson2004biological}
\begin{align}
C\left( {\left. {{\Omega _{{r_r}}},T} \right|\left\| {\textbf{x}} \right\|} \right) = \frac{1}{{{{\left( {4\pi DT} \right)}^{{3 \mathord{\left/
 {\vphantom {3 2}} \right.
 \kern-\nulldelimiterspace} 2}}}}}\exp \Big( { - \frac{{{{ \left\| {\textbf{x}} \right\| }^2}}}{{4DT}}} \Big). \label{Fraction_d_pass}
\end{align}

The fraction of  information molecules observed  inside the passive receiver with volume  ${V_{{\Omega _{{r_r}}}}}$ at time $T$  is denoted as
\begin{align}
{F^{\rm{PS}}}\left( {\left. {{\Omega _{{r_r}}},T} \right|{\left\| \textbf{x} \right\|}} \right) & = \int\limits_{{V_{{\Omega _{{r_r}}}}}} {C\left( {\left. {{\Omega _{{r_r}}},T} \right|\left\| {\textbf{x}} \right\|} \right) d{{V_{{\Omega _{{r_r}}}}}}} , \label{PS_nonuniform}
\end{align}
where ${V_{{\Omega _{{r_r}}}}}$ is the volume of the spherical passive receiver.
 
 % with ${V_{{\Omega _{{r_r}}}}} = {{4\pi {r_r}^3} \mathord{\left/{\vphantom {{4\pi {r_r}^3} 3}} \right.\kern-\nulldelimiterspace} 3}$

% \begin{align}
%{F^{\rm{PS}}}\left( {\left. {{\Omega _{{r_r}}},T} \right|{\left\| \textbf{x} \right\|}} \right) &  &= \int\limits_0^{{r_r}} {\int\limits_0^{2\pi } {\int\limits_0^\pi  {C\left( {\left. {{\Omega _{{r_r}}},T} \right|\left\| {\bf{x}} \right\|} \right){r^2}\sin \theta d\theta d\phi dr} } }, \label{PS_nonuniform}
%\end{align}

According to \eqref{PS_nonuniform} and Theorem 2 in \cite{noel2013using},   the fraction ${F^{\rm{PS}}}$ of  information molecules observed  inside the passive receiver at time $T$ due to a single  pulse emission by a transmitter at time $t$  is derived  as
\begin{align}
&{F^{\rm{PS}}}\left( {\left. {{\Omega _{{r_r}}},T} \right|{\left\| \textbf{x} \right\|}} \right)  =  \frac{1}{2}\left[ {{\rm{erf}}\Big( {\frac{{{r_r} - \left\| {\bf{x}} \right\|}}{{2\sqrt {DT} }}} \Big) + {\rm{erf}}\Big( {\frac{{{r_r} + \left\| {\bf{x}} \right\|}}{{2\sqrt {DT} }}} \Big)} \right] 
\nonumber \\ &+{\frac{{\sqrt{DT}}}{\sqrt{\pi }{\left\| {\bf{x}} \right\|}}} \left[ {\exp \Big( { - \frac{{{{\left( {{r_r} + \left\| {\bf{x}} \right\|} \right)}^2}}}{{4DT}}} \Big) - \exp \Big( { - \frac{{{{\left( {\left\| {\bf{x}} \right\| - {r_r}} \right)}^2}}}{{4DT}}} \Big)} \right], \label{PS_nonuniform_der}
\end{align}
which does \emph{not} assume that the molecule concentration inside the passive receiver is uniform.  This is 
unlike the common assumption that the concentration of molecule inside the passive receiver is uniform. Although that assumption is commonly applied, it relies on the receiver being sufficiently far from the transmitter (see  \cite{noel2013using}), which we cannot guarantee here since the transmitters are placed randomly.

%Similar as \cite{noel2014improving}, we assume that the information molecule inside the passive receiver is uniform and equal to that expected at the center of this receiver. Note that this assumption has been shown as accurate in the scenario with sufficiently large distance between transmitter and receiver \cite{noel2013using}. Thus, the expected number of observed molecule at the  receiver is converted from the expected concentration as 
%\begin{align}
%{\mathbb{E}}\left\{ {C\left( {\left. {{\Omega _{{r_r}}},t} \right|d} \right)} \right\} & =  C\left( {\left. {{\Omega _{{r_r}}},t} \right|d} \right){V_{{\Omega _{{r_r}}}}} 
%\nonumber\\&
%= \frac{{{N_{tx}}}}{{{{\left( {4\pi Dt} \right)}^{{3 \mathord{\left/
% {\vphantom {3 2}} \right.
% \kern-\nulldelimiterspace} 2}}}}}\exp \left( { - \frac{{{{\left( {d + {r_r}} \right)}^2}}}{{4Dt}}} \right){V_{{\Omega _{{r_r}}}}},
%\end{align}

In the large-scale molecular communication system with a passive receiver centered at the  origin, the expected number of  molecules observed inside the receiver at time $T$ due to a single  pulse emission at \emph{all} active transmitters at $t=0$ is given as 
\begin{align}
{\mathbb{E}}&\left\{ {N_{\rm{all}}^{{\rm{PS}}}\left( { {{\Omega _{{r_r}}},T} } \right)} \right\} = {\mathbb{E}}\left\{\sum\limits_{{\Phi _a}}{N_{\rm{tx}}^{\rm{PS}}}{F^{\rm{PS}}}\left( {\left. {{\Omega _{{r_r}}},T} \right|{\left\| \textbf{x} \right\|}} \right)\right\} 
\nonumber\\ & =\underbrace {N_{{\rm{tx}}}^{{\rm{PS}}}F_{\rm{u}}^{{\rm{PS}}}\left( {\left. {{\Omega _{{r_r}}},T} \right|\left\| {{{\bf{x}}^*}} \right\|} \right)}_{{\mathbb{E}}_{\rm{u}}^{{\rm{PS}}}} +\underbrace {N_{{\rm{tx}}}^{{\rm{PS}}}F_{\rm{I}}^{{\rm{PS}}}\left( {\left. {{\Omega _{{r_r}}},T} \right|\left\| {{{\bf{x}}}} \right\|} \right)}_{{\mathbb{E}}_{\rm{I}}^{{\rm{PS}}}}
,
\end{align}
where  ${F^{\rm{PS}}}\left( {\left. {{\Omega _{{r_r}}},T} \right|{\left\| \textbf{x} \right\|}} \right)$  is given in \eqref{PS_nonuniform_der},
%
%The SIR based on the expected number of  molecules observed inside the passive receiver at time $T$  is defined as 
%\begin{align}
%\rm{SIR{_{PS}}} = &\frac{{N_{\rm{tx}}^{\rm{PS}}}{\mathbb{E}}\left\{{F^{\rm{PS}}}\left( {\left. {{\Omega _{{r_r}}},T} \right|{\left\| {\textbf{x}}^* \right\|}} \right)\right\}}{{N_{\rm{tx}}^{\rm{PS}}}{\mathbb{E}}\Big\{\sum\limits_{{{{\Phi _a}} \mathord{\left/
% {\vphantom {{{\Phi _a}} {{{\textbf{x}}^*}}}} \right.
% \kern-\nulldelimiterspace} {{{\bf{x}}^*}}}}{F^{\rm{PS}}}\left( {\left. {{\Omega _{{r_r}}},T} \right|{\left\| \textbf{x} \right\|}} \right)\Big\}}= \frac{{{{\mathbb{E}}_{\rm{u}}^{{\rm{PS}}}}}}{{{{\mathbb{E}}_{\rm{I}}^{{\rm{PS}}}}}}, \label{SIRPS_raw}
%\end{align}
 ${{\mathbb{E}}_{\rm{u}}^{{\rm{PS}}}}$ is the fraction of  molecules observed inside the receiver at time $T$ due to the nearest transmitter, and ${{\mathbb{E}}_{\rm{I}}^{{\rm{PS}}}}$ is the fraction of  molecules observed inside the receiver at time $T$ due to the other (interfering) transmitters.

\section{Receiver Observations}

In this section, we first derive  the distance distribution between the receiver and the nearest point transmitter. By doing so, we derive  exact expressions for the expected number of  molecules observed inside  the receiver  due to the nearest point transmitter and that due to the interfering transmitters. %We also present  exact expressions for the SIRs. 

\subsection{Distance Distribution}
Unlike the  stochastic geometry modelling of wireless networks, where the transmitters are randomly located in the unbounded space, the point transmitters in a large-scale molecular communication system can only be distributed \emph{outside} the surface of the spherical receiver. Taking into account the minimum distance $r_r$ between point transmitters and the   receiver center, we derive the  probability density function (PDF) of the shortest distance between a point transmitter and the receiver in the following proposition.

\begin{ppro}
The PDF of the shortest distance between any point transmitter and the receiver  in 3D space  is given by 
\begin{align}
{f_{\left\| {{x^*}} \right\|}}(x) = 4\lambda_a\pi {x^2}{e^{ - \lambda_a \left( {\frac{4}{3}\pi {x^3} - \frac{4}{3}\pi {r_r}^3} \right)}}, \label{shortest_dist}
\end{align}
where $\lambda_a =\lambda \rho_{a}$.
\end{ppro}
\begin{proof}
See Appendix A.
\end{proof}

Based on the proof of Proposition 1, we also derive the PDF of the shortest distance between any point transmitter and the receiver  in 2D space in the following lemma.
\begin{cor}
The PDF of the shortest distance between any  point transmitter and the receiver  in 2D space  is given by 
\begin{align}
{f_{\left\| {{x^*}} \right\|}}(x) = 2\lambda_a \pi r{e^{ - \lambda_a \left( {\pi {r^2} - \pi {r_r}^2} \right)}}, \label{shortest_dist_2D}
\end{align}
where $\lambda_a =\lambda \rho_{a}$.
\end{cor}

\subsection{Absorbing Receiver Observations}
In this subsection, we derive a closed-form  expression for  the  expected  number of  molecules observed inside the absorbing receiver in 3D space.
 
 Using  Campbell’s theorem, we  derive  the expected number of absorbed molecules due to the  nearest transmitter   and  that due to the interfering transmitters until time $t$ as
\begin{align}
{{{\mathbb{E}}_{\rm{u}}^{{\rm{FA}}}}}= &{N_{\rm{tx}}^{\rm{FA}}}\int_{{r_r}}^\infty  {{\Pi_0^t} \left( x \right)4{\lambda _a}\pi {x^2}}  
{e^{ - {\lambda_a} \left( {\frac{4}{3}\pi {x^3} - \frac{4}{3}\pi {r_r}^3} \right)}}dx
, \label{E_u_FA}
\end{align}
and
\begin{align}
\hspace{-0.2cm}{{{\mathbb{E}}_{\rm{I}}^{{\rm{FA}}}}} =& {N_{\rm{tx}}^{\rm{FA}}}\left(4\pi {\lambda _a} \right)^2{  \int_{{r_r}}^\infty  {\int_x^\infty  {\Pi_0^t} \left( r \right)} } {r^2}dr
 {x^2}{e^{ - {\lambda _a}\left( {\frac{4}{3}\pi {x^3} - \frac{4}{3}\pi {r_r}^3} \right)}}dx, \label{E_I_FA}
\end{align}
respectively.

\begin{theorem}
The   expected  net number of  molecules absorbed at the  absorbing receiver in 3D space during any sampling time interval $[t,t+T_{ss}]$ is derived as
\begin{align}
{\mathbb{E}}&\left\{ N_{{\rm{all}}}^{{\rm{FA}}}\left( { {{\Omega _{{r_r}}},t,t + {T_{ss}}} } \right) \right\}  
\nonumber\\&= 4{N_{{\rm{tx}}}^{\rm{FA}}}\sqrt \pi  {\lambda _a}{r_r}\left[ {D\sqrt \pi  {T_{ss}} + 2\sqrt D {r_r}\left( {\sqrt {{T_{ss}} + t}  - \sqrt t } \right)} \right] .  \label{Total_FA_final}
\end{align}
The   expected   number of  molecules observed inside the fully absorbing receiver in 3D space until time $t$ is derived as
\begin{align}
{\mathbb{E}}&\left\{ N_{{\rm{all}}}^{{\rm{FA}}}\left( { {{\Omega _{{r_r}}},t} } \right) \right\}  
= 4{N_{{\rm{tx}}}^{\rm{FA}}}\sqrt \pi  {\lambda _a}{r_r}\left[ {D\sqrt \pi  t + 2{r_r}\sqrt {D t} } \right] .  \label{Total_FA_final_t}
\end{align}
%By substituting $T=0$ into \eqref{Total_FA_final}, we derive the cumulative total  number of adsorbed molecule until time $T$ as
%\begin{align}
%{\mathbb{E}}&\left\{ F_{{\rm{tot}}}^{{\rm{FA}}}\left( {\left. {{\Omega _{{r_r}}},T} \right|\left\| \textbf{x}^* \right\|} \right) \right\}  
%\nonumber\\&= {N_{{\rm{tx}}}^{\rm{FA}}}4\sqrt \pi  {\lambda _a}{r_r}\left[ {D\sqrt \pi  {T_{ss}} + \sqrt D {r_r}2\left( {\sqrt {{T_{ss}}} } \right)} \right].  \label{Total_FA_final}
%\end{align}
\end{theorem}
\begin{proof}
See Appendix B.
 \end{proof}
 
 From Theorem 1,  we find that the  expected number of  molecules absorbed at the absorbing receiver  at time $t$ is linearly proportional to the density of active transmitters, and increases with increasing  diffusion coefficient or receiver radius. As expected, we find that the expected number of  molecules absorbed at the receiver until  $ t$  is always increasing with time $ t$.

% \begin{ppro}
%The SIR of the expected number of molecules observed inside the absorbing receiver  at time $t$  in 3D space is derived as
%\begin{align}
%\rm{SIR}_{FA} = &\frac{{\int_{{r_r}}^\infty  {{\Pi_0^t} \left( x \right){x^2}{e^{ - \frac{4}{3}\pi {\lambda _a}{x^3}}}dx} }}{{\left( {4\pi {\lambda _a}} \right)\int_{{r_r}}^\infty  {\left( {\int_x^\infty  {{\Pi_0^t} \left( r \right){r^2}dr} } \right)} {x^2}{e^{ - \frac{4}{3}\pi {\lambda _a}{x^3}}}dx}}, \label{SIR_FA}
%\end{align}
%where $\Pi_0 \left( r \right)$ is given as
%\begin{align}
%&  {\Pi_0^t} \left( r\right) =  \int_0^{t } {\frac{{{r_r}}}{r}\frac{{r - {r_r}}}{{\sqrt {4\pi D{t^3}} }}\exp \Big( { - \frac{{{{\left( {r - {r_r}} \right)}^2}}}{{4Dt}}} \Big)dt} , \label{phi}
%\end{align}
%\end{ppro}
%\begin{proof}
%See Appendix C.
%\end{proof}
%where 
%\begin{align}
%{{F_{\rm{u}}^{{\rm{FA}}}}}= &\int_{{r_r}}^\infty  {{N_{hit}}\left( {\left. {{\Omega _{{r_r}}},T,T + {T_{ss}}} \right|\left\| {\bf{x}} \right\|} \right)4{\lambda _a}\pi {x^2}}  
%\nonumber\\& {e^{ - {\lambda_a} \left( {\frac{4}{3}\pi {x^3} - \frac{4}{3}\pi {r_r}^3} \right)}}dx
%, \label{E_u_FA}
%\end{align}
%\begin{align}
%{{{F}_{\rm{I}}^{{\rm{FA}}}}} =& 4\pi {\lambda _a} {  \int_{{r_r}}^\infty  {\int_x^\infty  {{N_{hit}}\left( {\left. {{\Omega _{{r_r}}},T,T + {T_{ss}}} \right|\left\| {\bf{x}} \right\|} \right)} } } 
%\nonumber\\& 4{\lambda _a}\pi {x^2}{e^{ - {\lambda _a}\left( {\frac{4}{3}\pi {x^3} - \frac{4}{3}\pi {r_r}^3} \right)}}{r^2}drdx, \label{E_I_FA}
%\end{align}
%and

\subsection{Passive Receiver Observations}
In the following theorem, we derive   the expected  number of   molecules   observed inside the passive receiver in 3D space.

Using  Campbell’s theorem, we  derive  the expected number of observed molecules due to the  nearest transmitter   and  that due to the interfering transmitters at time $t$ as
\begin{align}
&{{{\mathbb{E}}_{\rm{u}}^{{\rm{PS}}}}}= 4{\lambda _a}\pi N_{{\rm{tx}}}^{{\rm{PS}}}{e^{\frac{4}{3}\pi {r_r}^3{\lambda _a}}}\int_{{r_r}}^\infty  {\Phi \left( x \right){x^2}\exp \Big\{ { - \frac{4}{3}\pi {x^3}{\lambda _a}} \Big\}dx} 
, \label{E_u_PS}
\end{align}
and
\begin{align}
{{{\mathbb{E}}_{\rm{I}}^{{\rm{PS}}}}} =&  {\left( {4\pi {\lambda _a}} \right)^2}{e^{\frac{4}{3}\pi {r_r}^3{\lambda _a}}}N_{{\rm{tx}}}^{\rm{PS}}\int_{{r_r}}^\infty  {\int_x^\infty  {\Phi \left( r \right)} }{r^2}dr 
 \nonumber\\& {x^2}{e^{ - \frac{4}{3}\pi {x^3}{\lambda _a}}}dx, \label{E_I_PS}
\end{align}
respectively. In \eqref{E_u_PS} and \eqref{E_I_PS}, ${\Phi \left( r \right)} ={F^{\rm{PS}}}\left( {\left. {{\Omega _{{r_r}}},t} \right|r} \right)$.

\begin{theorem}
The  expected net number of  molecules observed inside the passive receiver during any sampling time interval $[t,t+T_{ss}]$ in 3D space  is derived as
\begin{align}
&{\mathbb{E}}\left\{ {N_{{\rm{all}}}^{{\rm{PS}}}\left( {\left. {{\Omega _{{r_r}}},t,t + {T_{ss}}} \right|\left\| {\bf{x}} \right\|} \right)} \right\}= 4N_{\rm{tx}}^{\rm{PS}}\pi {\lambda _a}
\nonumber\\ &\left[ {\int_{{r_r}}^\infty  {{F^{{\rm{PS}}}}\left( {\left. {{\Omega _{{r_r}}},t + {T_{ss}}} \right|r} \right)} {r^2}dr} \right.
\left. { - \int_{{r_r}}^\infty  {{F^{{\rm{PS}}}}\left( {\left. {{\Omega _{{r_r}}},t} \right|r} \right)} {r^2}dr} \right]
, \label{Total_PS_final}
\end{align}
where   ${F^{\rm{PS}}}\left( {\left. {{\Omega _{{r_r}}},t} \right|{r}} \right)$  is given in \eqref{PS_nonuniform_der}.

The  expected  number of  molecules observed inside the passive receiver at time $t$ in 3D space  is derived as
\begin{align}
&{\mathbb{E}}\left\{ {N_{{\rm{all}}}^{{\rm{PS}}}\left( {\left. {{\Omega _{{r_r}}},0,t } \right|\left\| {\bf{x}} \right\|} \right)} \right\}= 
\nonumber\\ & \hspace{2cm}4N_{\rm{tx}}^{\rm{PS}}\pi {\lambda _a}{\int_{{r_r}}^\infty  {{F^{{\rm{PS}}}}\left( {\left. {{\Omega _{{r_r}}},t } \right|r} \right)} {r^2}dr} 
. \label{Total_PS_final_t}
\end{align}

\end{theorem}
\begin{proof}
Analogous to Appendix B without solving the integrals.
\end{proof}
In Theorem 2, we observe that the expected number of  molecules observed inside the passive receiver also increases proportionately with the density of active transmitters.

%\begin{ppro}
%The  SIR of the expected number of molecules observed inside the passive receiver at time $t$ in 3D space  is derived as
%\begin{align}
%\rm{SIR}_{PS} = &\frac{{\int_{{r_r}}^\infty  {{F^{{\rm{PS}}}}\left( {\left. {{\Omega _{{r_r}}},t} \right|x} \right){x^2}e^{ { - \frac{4}{3}\pi {x^3}{\lambda _a}} }dx} }}{{4\pi {\lambda _a}\smallint _{{r_r}}^\infty \left( \int_x^\infty{F^{{\rm{PS}}}}\left( {\left. {{\Omega _{{r_r}}},t} \right|r} \right){r^2}dr \right){x^2}e^{ { - \frac{4}{3}\pi {x^3}{\lambda _a}} }dx}}, \label{SIR_PS}
%\end{align}
%where ${F^{\rm{PS}}}\left( {\left. {{\Omega _{{r_r}}},t} \right|r} \right)$  is given in \eqref{PS_nonuniform_der}.
%
%\end{ppro}
%\begin{proof}
%See Appendix D.
%\end{proof}

\section{Numerical and Simulation Results}

In this section, we examine the expected number of molecules observed at the absorbing receiver and the passive receiver due to simultaneous single pulse emissions at all active point transmitters. In all figures of this section,  we set the parameters as follows: $r_r = 5\,\micron$ and $N_{\rm{tx}}^{\rm{FA}}= N_{\rm{tx}}^{\rm{PS}}=10^4$.  
In all figures,  the analytical curves of the expected number of  molecules absorbed at the absorbing receiver 
due to all the transmitters,  the nearest transmitter, and the interfering transmitters are plotted using Eqs. \eqref{Total_FA_final}, \eqref{E_u_FA},  and \eqref{E_I_FA}, and are abbreviated as “Absorbing All”, “Absorbing Nearest”, and “Absorbing Aggregate”, respectively. The analytical curves of the expected number of  molecules observed inside   the passive receiver due to all the transmitters, the nearest transmitter, the interfering transmitters  are plotted  using \eqref{Total_PS_final}, \eqref{E_u_PS},  and \eqref{E_I_PS}, and are abbreviated as “Passive All”, “Passive Nearest”, and “Passive Aggregate”, respectively.   
%In Fig. \ref{Fig3}, we also plot the analytical curves of the SIRs at the absorbing receiver and the passive receiver  using Eqs. \eqref{SIR_FA} and \eqref{SIR_PS}, which are abbreviated as   “Absorbing SIR” and “Passive SIR”, respectively.

\subsection{Particle-Based  and Pseudo Simulation Validation}

In Fig. \ref{fig_accord_normalized}, we set   $D = 80\times10^{-12} \frac{\meter^2}{\second}$, and assume that the transmitters are placed up to $R = 50\,\micron$ from the center of the receiver at a density of $\lambda_a = 10^{-4}$ transmitters per $\micron^3$ (i.e.,  52 average number of transmitters, including the subtraction of the receiver volume).
 The receiver takes samples every $T_{ss} = 0.01\; \second$ and calculates the net change in the number of observed molecules between samples. The default simulation time step is also $0.01$ $\rm{\second}$. Unless otherwise noted, all simulation results were averaged over $10^4$ transmitter location permutations, with each permutation simulated at least 10 times.

In Fig. \ref{fig_accord_normalized}, we verify the analytical expressions for the expected net number of molecules observed during $[t,t+T_{ss}]$ at the absorbing receiver in Eq. \eqref{E_u_FA} and Eq. \eqref{E_I_FA}, and that inside the passive receiver in Eq. \eqref{E_u_PS} and Eq. \eqref{E_I_PS}  by comparing with the particle-based simulations and the Monte Carlo simulations. The  particle-based  simulations were performed by tracking the progress
of individual particles to obtain the net number of  observed molecules     during $[t,t+T_{ss}]$ using the AcCoRD simulator (Actor-based Communication via Reaction-Diffusion) \cite{noel2016_simulator}.  The pseudo simulations rely on the Monte Carlo simulation method, which were performed   by  averaging the  expected number of observed molecules  due to  all active transmitters with  randomly-generated location, as calculated from Eq. \eqref{Fraction_x} and Eq. \eqref{PS_nonuniform_der}, over $10^4$ realizations.

 In the right subplot of Fig.~\ref{fig_accord_normalized}, we compare passive and absorbing receivers and observe the expected  net number of observed molecules during $[t,t+T_{ss}]$ due to the nearest transmitter and due to the aggregation of the interfering transmitters. In the left subplot of Fig.~\ref{fig_accord_normalized}, we lower the simulation time step to $10^{-4}\,\second$ for the first few samples of the two absorbing receiver cases, in order to demonstrate the corresponding improvement in accuracy. All curves in both subplots are scaled by the maximum value of the corresponding analytical curve in the right subplot; the scaling values and other simulation parameters are summarized in Table~\ref{table_accord}.

\begin{figure}[!tb]
	\centering
	\includegraphics[width=\linewidth]{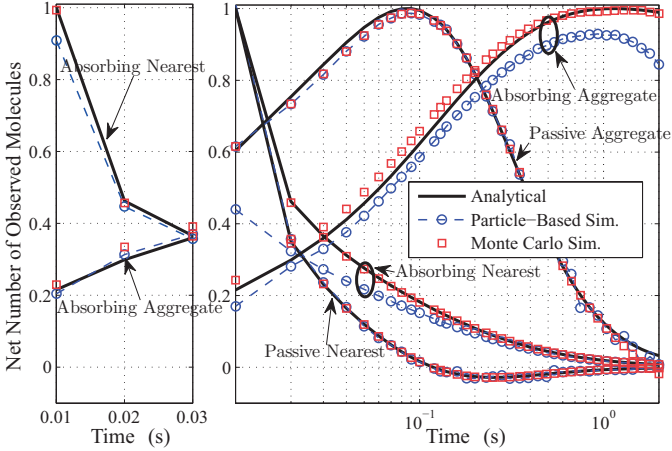}
	\caption{Net number of observed molecules inside the receiver as a function of time. All curves are scaled by the maximum value of the analytical curves in the right subplot.}
	\label{fig_accord_normalized}
\end{figure}

\begin{table}[!tb]
	\centering
	\caption{ The simulation parameters and scaling values applied in Fig.~\ref{fig_accord_normalized}.}
	{\renewcommand{\arraystretch}{0.8}
		\begin{tabular}{|c|c|c|c|c|}
			\hline
			 Transmitter & Receiver & Realizations & Time  & Scaling \\
			&  &  & Step [s] &Value \\ \hline
			Nearest & Passive & $10^4$ & $10^{-2}$ & 149.57 \\ \hline
			Nearest & Active & $10^4$ & $10^{-2}$ & 354.52 \\ \hline
			Aggregate & Passive & $10^4$ & $10^{-2}$ & 9.252 \\ \hline
			Aggregate & Active & $10^3$ & $10^{-3}$ & 59.42 \\ \hline
		\end{tabular}
	}
	\label{table_accord}
\end{table}

\subsubsection{Particle-Based  Simulation Validation}
Overall, there is good agreement between the analytical curves and the particle-based simulations in the right subplot of Fig.~\ref{fig_accord_normalized}. The analytical results for the net number of molecules observed  inside the passive receiver  during $[t,t+T_{ss}]$ due to  the nearest transmitter is highly accurate, and even captures the net loss of molecules observed after $t = 0.1\,\second$. There is a slight deviation in the particle-based simulation for the ``passive aggregate" curve, particularly as time approaches $t = 1\,\second$, which is primarily due to the very low number of molecules observed at this time (note that the scaling factor in this case is only $9.252$; see Table~\ref{table_accord}).

 There is less agreement between the particle-based simulations and the analytical expressions for the absorbing receiver, and this is primarily due to the large simulation time step (even though we used a smaller time step for the aggregate transmitter case in the right subplot; see Table II). To demonstrate the impact of the time step, the left subplot shows much better agreement for the absorbing receiver model by lowering the time step to $10^{-4}\,\second$. This improvement is especially true in the case of the nearest transmitter, as there is significant deviation between the particle-based simulation and the analytical expression for very early times in the right subplot.

%even though the simulations restricted the placement of transmitters to the maximum distance $R = 50\,\micron$

\subsubsection{Monte Carlo Simulation Validation}
There is a good match between the analytical curves and the Monte Carlo simulations for the  net number of  molecules observed at both types of receiver during $[t,t+T_{ss}]$ due to the nearest transmitter, which can be attributed to the large number of  molecules and the shortest distance value compared with $R = 50\,\micron$ (as shown in Table~\ref{table_accord}). There is slight deviation in the Monte Carlo simulations for the expected number of  molecules observed inside both types of receiver due to the interfering transmitters, and this is primarily due to
  the restricted placement of transmitters to the maximum distance $R = 50\,\micron$. In Figs. \ref{Fig2},  better agreement between the analytical curves and pseudo simulation is achieved by increasing  the maximum placement distance  $R$.

Due to the extensive computational demands to simulate such large molecular communication environments, we assume that the particle-based simulations have sufficiently verified the analytical models. The remaining simulation results in Fig. \ref{Fig2}   is only  generated via Monte Carlo simulation.

\begin{figure}[!tb]
	\centering
	\includegraphics[width=\linewidth]{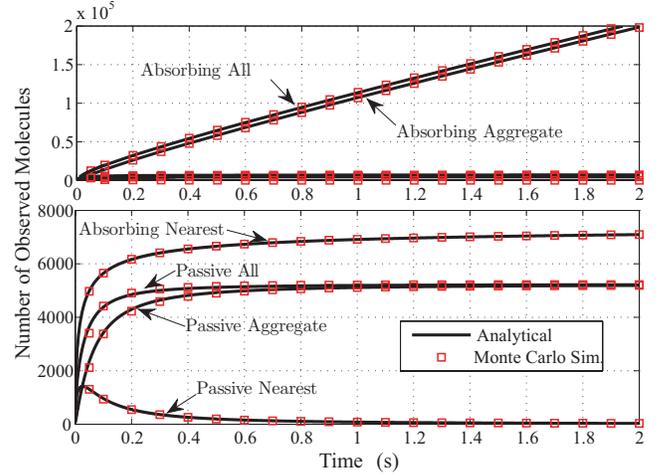}
	\caption{Expected number of  molecules observed inside the receiver  as a function of time. }
	\label{Fig2}
\end{figure}

\subsection{Performance Evaluation}
From Fig.~\ref{fig_accord_normalized} and  the scaling values in Table II,  we see that the expected net number of  molecules absorbed at  the absorbing receiver  is much larger than that inside the passive receiver, since every molecule arriving at the absorbing receiver is permanently absorbed.  We also notice that  the expected  net number of absorbed molecules due to the nearest transmitter  is much larger than that due to the  interfering transmitters, which may be due to a relatively low transmitter density.

Interestingly, the concurrent single pulse transmission by the transmitters at time $t=0$ results in a longer and stronger channel response at the absorbing receiver than that at the passive receiver. If the demodulation is based on the number of observed molecules during each bit interval, the longer channel response at the absorbing receiver may contribute to higher ISI  than at a passive receiver for the same bit interval,  whereas its stronger channel response may benefit  signal detection.

In Fig.~\ref{Fig2}, we set the parameters: $D = 120\times10^{-12} \frac{\meter^2}{\second}$, $R = 100\,\micron$, and $T_{ss} = 0.1\; \second$.
Fig.~\ref{Fig2}  plots the expected number of  molecules observed at the absorbing receiver and the passive receiver at time $t$. We set the density of active transmitters as $\lambda_a = 10^{-3} /\micron^{3}$.  
As shown in the lower subplot of  Fig.~\ref{Fig2},  the   channel responses of the receivers due to the nearest transmitter in this large-scale molecular communication system  are consistent with those observed at the absorbing receiver in   \cite[Fig. 4]{Yansha16} and the passive receiver in  \cite[Fig. 2]{Llatser13} and \cite[Fig. 1]{noel2014improving}   for a point-to-point molecular communication system.

In Fig.~\ref{Fig2}, we notice that the  expected   number of observed molecules at time $t$ due to all the transmitters is dominated by the interfering transmitters, rather than the nearest transmitter, which is due to the higher density of transmitters.
Furthermore, as we might expect, the  expected  number of  molecules observed inside the passive receiver at time $t$ stabilizes after $t=0.8\,\second$, whereas that at the absorbing receiver increases linearly with increasing time. This reveals the potential differences in optimal demodulation and interference cancellation design for these two types of receiver.

%Fig.~\ref{Fig3}  plots the expected number of  molecules observed inside the absorbing receiver and the passive receiver at $t=2\,$s versus the density of active transmitters $\lambda_a$.   With the increase of $\lambda_a$, 
% the  number of observed molecules due to the interfering transmitters increases,  whereas the  number of observed molecules due to the nearest transmitter remains almost unchanged, which results in the decrease of SIRs. Interestingly, the SIR at the absorbing receiver is always higher than that at the passive receiver, which indicates the potential higher reliability  at the absorbing receiver. This observation showcases the potential benefits brought by using an absorbing receiver in a large-scale molecular communication system.

%\begin{figure}[!tb]
%	\centering
%	\includegraphics[width=\linewidth]{Fig2_FA_Pas_T_b_linear}
%	\caption{Expected number of  molecules observed inside the receiver at time $t=2\;\;\second$ as a function of the density of transmitters. }
%	\label{Fig3}
%\end{figure}
\section{Conclusions and Future Work}
In this paper, we provided a general model for the transmitter modelling in a large-scale molecular communication system using stochastic geometry.
The  collective signal strength at a fully absorbing receiver and a passive receiver are modelled and examined. We derived  tractable  expressions for the expected  number of observed molecules at the fully absorbing receiver and the passive receiver, which were shown to increase with transmitter density.
%We also defined and derived  analytical expressions for the  SIRs at the receiver. 
Our analytical results were validated through  particle-based simulation and Monte Carlo simulation.
%It is shown that the SIR of a fully absorbing receiver is greater than  that of a passive receiver, which indicates the potential higher reliability  of the fully absorbing receiver. 
The analytical model presented in this paper can also be applied for the performance evaluation of other types of receiver (e.g., partially absorbing, reversible adsorption receiver, ligand-binding receiver)   in large-scale molecular communication systems by substituting its corresponding channel response.

\appendices
\numberwithin{equation}{section}

\bibliographystyle{vancouver}
\bibliography{Ref}

\begin{thebibliography}{10}
\providecommand{\url}[1]{#1}
\csname url@samestyle\endcsname
\providecommand{\newblock}{\relax}
\providecommand{\bibinfo}[2]{#2}
\providecommand{\BIBentrySTDinterwordspacing}{\spaceskip=0pt\relax}
\providecommand{\BIBentryALTinterwordstretchfactor}{4}
\providecommand{\BIBentryALTinterwordspacing}{\spaceskip=\fontdimen2\font plus
\BIBentryALTinterwordstretchfactor\fontdimen3\font minus
  \fontdimen4\font\relax}
\providecommand{\BIBforeignlanguage}[2]{{%
\expandafter\ifx\csname l@#1\endcsname\relax
\typeout{** WARNING: IEEEtran.bst: No hyphenation pattern has been}%
\typeout{** loaded for the language `#1'. Using the pattern for}%
\typeout{** the default language instead.}%
\else
\language=\csname l@#1\endcsname
\fi
#2}}
\providecommand{\BIBdecl}{\relax}
\BIBdecl

\bibitem{EckfordBook13}
T.~Nakano, A.~Eckford, and T.~Haraguchi, \emph{Molecular communication}.\hskip
  1em plus 0.5em minus 0.4em\relax Cambridge University Press, 2013.

\bibitem{Codling08}
E.~Codling, M.~Plank, and S.~Benhamous, ``Random walk models in biology,''
  \emph{Journal of The Royal Society Interface}, vol.~5, no.~25, pp. 813--834,
  Aug. 2008.

\bibitem{Atkinson09}
S.~Atkingson and P.~Williams, ``Quorum sensing and social networking in the
  microbial world,'' \emph{Journal of The Royal Society Interface}, vol.~6,
  no.~40, pp. 959--978, Aug. 2009.

\bibitem{Llatser13}
I.~Llatser, A.~Cabellos-Aparicio, and M.~Pierobon, ``Detection techniques for
  diffusion-based molecular communication,'' \emph{IEEE Journal on Selected
  Areas in Communications (JSAC)}, vol.~31, pp. 726--734, Dec. 2013.

\bibitem{Guo15TMBMC}
W.~Guo, C.~Mias, N.~Farsad, and J.~Wu, ``Molecular versus electromagnetic wave
  propagation loss in macro-scale environments,'' \emph{IEEE Trans. Mol. Biol.
  Multi-Scale Commun.}, vol.~1, Mar. 2015.

\bibitem{Douglas12}
S.~M. Douglas, I.~Bachelet, and G.~M. Church, ``A logic-gated nanorobot for
  targeted transport of molecular payloads,'' \emph{Science}, vol. 335, no.
  6070, pp. 831--834, Feb. 2012.

\bibitem{Cavalcanti06}
A.~Cavalcanti, T.~Hogg, B.~Shirinzadeh, and H.~Liaw, ``Nanorobot communication
  techniques: a comprehensive tutorial,'' in \emph{Proc. IEEE Int. Conf.
  Control, Autom., Robot., Vis.}, Dec. 2006, pp. 1--6.

\bibitem{Kirkpatrick10}
C.~J. Kirkpatrick and W.~Bonfield, ``Nanobiointerface: a multidisciplinary
  challenge,'' \emph{Journal of The Royal Society Interface}, vol.~7, no. Suppl
  1, pp. S1--S4, Dec. 2009.

\bibitem{noel2014unify}
A.~Noel, K.~C. Cheung, and R.~Schober, ``A unifying model for external noise
  sources and isi in diffusive molecular communication,'' \emph{{IEEE} J. Sel.
  Areas Commun.}, vol.~32, no.~12, pp. 2330--2343, Dec 2014.

\bibitem{yilmaz2014simulation}
H.~B. Yilmaz and C.-B. Chae, ``Simulation study of molecular communication
  systems with an absorbing receiver: Modulation and {ISI} mitigation
  techniques,'' \emph{Simulat. Modell. Pract. Theory}, vol.~49, pp. 136--150,
  Dec. 2014.

\bibitem{Yansha16}
\BIBentryALTinterwordspacing
Y.~Deng, A.~Noel, M.~Elkashlan, A.~Nallanathan, and K.~C. Cheung, ``Modeling
  and simulation of molecular communication systems with a reversible
  adsorption receiver,'' \emph{arXiv}, 2016. [Online]. Available:
  \url{http://arxiv.org/abs/1601.00681}
\BIBentrySTDinterwordspacing

\bibitem{pierobon2014statistical}
M.~Pierobon and I.~F. Akyildiz, ``A statistical--physical model of interference
  in diffusion-based molecular nanonetworks,'' \emph{{IEEE} Trans. Commun.},
  vol.~62, no.~6, pp. 2085--2095, Jun. 2014.

\bibitem{baccelli2009stochastic}
F.~Baccelli and B.~Blaszczyszyn, \emph{Stochastic geometry and wireless
  networks: Volume 1: Theory}.\hskip 1em plus 0.5em minus 0.4em\relax Now
  Publishers Inc, 2009, vol.~1.

\bibitem{yan2016sensor}
Y.~Deng, L.~Wang, M.~Elkashlan, A.~Nallanathan, and R.~K. Mallik, ``Physical
  layer security in three-tier wireless sensor networks: {A} stochastic
  geometry approach,'' \emph{IEEE Trans. Inf. Forensics Security}, vol.~11,
  no.~6, pp. 1128--1138, Jun. 2016.

\bibitem{yan16hetnet}
Y.~Deng, L.~Wang, M.~Elkashlan, M.~Direnzo, and J.~Yuan, ``Modeling and
  analysis of wireless power transfer in heterogeneous cellular networks,''
  \emph{{IEEE} Trans. Commun.}, 2016.

\bibitem{Jeanson11}
S.~Jeanson, J.~Chadoeuf, M.~Madec, S.~Aly, J.~Floury, T.~F. Brocklehurst, and
  S.~Lortal, ``Spatial distribution of bacterial colonies in a model cheese,''
  \emph{Applied and Environmental Microbiology}, vol.~77, no.~4, pp.
  1493--1500, Dec. 2010.

\bibitem{Yilmaz14}
H.~B. Yilmaz, A.~C. Heren, T.~Tugcu, and C.-B. Chae, ``{Three-Dimensional
  channel characteristics for molecular communications with an absorbing
  receiver},'' \emph{IEEE Communications Letters}, vol.~18, no.~6, pp.
  929--932, Jun. 2014.

\bibitem{noel2014improving}
A.~Noel, K.~C. Cheung, and R.~Schober, ``Improving receiver performance of
  diffusive molecular communication with enzymes,'' \emph{IEEE Trans.
  Nanobiosci.}, vol.~13, no.~1, pp. 31--43, Mar. 2014.

\bibitem{cussler2009diffusion}
E.~L. Cussler, \emph{Diffusion: mass transfer in fluid systems}.\hskip 1em plus
  0.5em minus 0.4em\relax Cambridge university press, 2009.

\bibitem{nelson2004biological}
P.~Nelson, \emph{Biological Physics: Energy, Information, Life}, updated
  1st~ed.\hskip 1em plus 0.5em minus 0.4em\relax W. H. Freeman and Company,
  2008.

\bibitem{noel2013using}
A.~Noel, K.~C. Cheung, and R.~Schober, ``Using dimensional analysis to assess
  scalability and accuracy in molecular communication,'' in \emph{Proc. IEEE
  ICC MoNaCom}, Jun. 2013, pp. 818--823.

\bibitem{noel2016_simulator}
\BIBentryALTinterwordspacing
A.~Noel. (2016) Actor-based communication via reaction-diffusion. [Online].
  Available: \url{https://github.com/adamjgnoel/AcCoRD}
\BIBentrySTDinterwordspacing

\end{thebibliography}

\balance
\end{document}